\documentclass[10pt,a4paper]{article}

\usepackage[utf8]{inputenc}
\usepackage[T1]{fontenc}
\usepackage[table]{xcolor}
\definecolor{primary}{HTML}{006BA2}
\definecolor{accent}{HTML}{E3120B}
\definecolor{neutral}{HTML}{4A4A4A}
\definecolor{boxfill}{HTML}{E8F4FD}
\definecolor{boxborder}{HTML}{006BA2}
\definecolor{defcolor}{HTML}{2E7D32}
\definecolor{todobg}{HTML}{FFF4E5}
\definecolor{todoframe}{HTML}{B45309}

\usepackage{style/arxiv-paper}
\usepackage{hyperref}
\usepackage{cleveref}
\usepackage{thmtools}
\usepackage{authblk}
\usepackage{orcidlink}

\hypersetup{
  colorlinks=true,
  linkcolor=primary,
  citecolor=primary,
  urlcolor=accent,
  bookmarksnumbered=true,
  bookmarksopen=true,
  pdflang={en-US}
}

\declaretheoremstyle[
  headfont=\bfseries\color{primary},
  notefont=\normalfont,
  bodyfont=\itshape,
  headpunct={.},
  spaceabove=0pt,
  spacebelow=0pt
]{thmplain}
\declaretheoremstyle[
  headfont=\bfseries\color{defcolor},
  notefont=\normalfont,
  bodyfont=\normalfont,
  headpunct={.},
  spaceabove=0pt,
  spacebelow=0pt
]{thmdef}
\declaretheoremstyle[
  headfont=\itshape\bfseries\color{neutral},
  notefont=\normalfont,
  bodyfont=\normalfont,
  headpunct={.},
  spaceabove=0pt,
  spacebelow=0pt
]{thmrem}

\declaretheorem[style=thmplain,name=Theorem,numberwithin=section]{theorem}
\declaretheorem[style=thmplain,name=Lemma,sibling=theorem]{lemma}

\declaretheorem[style=thmplain,name=Corollary,sibling=theorem]{corollary}

\declaretheorem[style=thmdef,name=Assumption,sibling=theorem]{assumption}
\declaretheorem[style=thmdef,name=Open Problem]{openproblem}
\declaretheorem[style=thmrem,name=Remark,sibling=theorem]{remark}

\tcbset{
  thmbox/.style={
    enhanced,sharp corners,boxrule=0pt,frame hidden,
    colback=primary!6,borderline west={2.5pt}{0pt}{primary},
    left=10pt,right=8pt,top=6pt,bottom=6pt,
    before skip=8pt,after skip=8pt
  },
  defbox/.style={
    enhanced,sharp corners,boxrule=0pt,frame hidden,
    colback=defcolor!7,borderline west={2.5pt}{0pt}{defcolor},
    left=10pt,right=8pt,top=6pt,bottom=6pt,
    before skip=8pt,after skip=8pt
  },
  rembox/.style={
    enhanced,sharp corners,boxrule=0pt,frame hidden,
    colback=neutral!5,borderline west={2.5pt}{0pt}{neutral},
    left=10pt,right=8pt,top=6pt,bottom=6pt,
    before skip=8pt,after skip=8pt
  },
  openbox/.style={
    enhanced,sharp corners,boxrule=0pt,frame hidden,
    colback=defcolor!7,borderline west={2.5pt}{0pt}{defcolor},
    left=8pt,right=7pt,top=3pt,bottom=3pt,
    before skip=4pt,after skip=4pt
  }
}
\tcolorboxenvironment{theorem}{thmbox}
\tcolorboxenvironment{lemma}{thmbox}
\tcolorboxenvironment{proposition}{thmbox}
\tcolorboxenvironment{corollary}{thmbox}
\tcolorboxenvironment{conjecture}{thmbox}
\tcolorboxenvironment{definition}{defbox}
\tcolorboxenvironment{assumption}{defbox}
\tcolorboxenvironment{openproblem}{openbox}
\tcolorboxenvironment{remark}{rembox}

\newcommand{\doenotice}{%
This manuscript has been authored by UT-Battelle, LLC under Contract
No.\ DE-AC05-00OR22725 with the U.S. Department of Energy. The publisher,
by accepting the article for publication, acknowledges that the United
States Government retains a non-exclusive, paid-up, irrevocable, world-wide
license to publish or reproduce the published form of this manuscript, or
allow others to do so, for United States Government purposes. The Department
of Energy will provide public access to these results of federally sponsored
research in accordance with the DOE Public Access Plan
(\url{https://www.energy.gov/doe-public-access-plan}).%
}

\author[1]{Piyush Sao\,\orcidlink{0000-0002-9432-5855}%
  \thanks{Corresponding author: \href{mailto:saopk@ornl.gov}{saopk@ornl.gov}}}
\author[1]{Narasinga Miniskar\,\orcidlink{0000-0001-8259-8891}}
\author[1]{Pedro Valero-Lara\,\orcidlink{0000-0002-1479-4310}}
\author[1]{Keita Teranishi\,\orcidlink{0000-0001-6647-2690}}
\author[1]{Sudip~Seal\,\orcidlink{0000-0003-3233-0656}}
\affil[1]{Oak Ridge National Laboratory, Oak Ridge, Tennessee 37831, USA\\
  \texttt{saopk@ornl.gov, miniskarnr@ornl.gov, valerolarap@ornl.gov}\\
  \texttt{teranishik@ornl.gov, sealsk@ornl.gov}}
\date{}

\title{\bfseries A Nuclear-Norm Lower Bound for\\
Dithered Scalar Quantization of Matrix Products\thanks{\doenotice}}
\hypersetup{
  pdftitle={A Nuclear-Norm Lower Bound for Dithered Scalar Quantization of Matrix Products},
  pdfauthor={Piyush Sao, Narasinga Miniskar, Pedro Valero-Lara, Keita Teranishi, Sudip Seal},
  pdfsubject={Lower bounds and gauge transformations for scalar-quantized matrix products},
  pdfkeywords={matrix quantization, nuclear norm, dither, matrix products, gauge transformations}
}

\begin{document}
\maketitle

\begin{abstract}
When computing a low-precision matrix product \(C=AB\), scalar quantization
rounds the entries of the two input factors.  The resulting error is sensitive
to how \(C\) is factorized: a product-preserving transformation can change the
largest absolute entries---the \emph{ranges}---of the factor rows and columns.
These ranges set the quantization grid steps, so the transformation can change
the error without changing \(C\).  We ask for the smallest leading
contribution to the expected squared error over all invertible inner changes
of basis and orthogonal outer rotations.  Under
independent, zero-mean subtractive-dither noise on an unbounded lattice, we
derive an output-only lower bound for the leading expected squared error:
\[
  E_{\mathrm{lead}}\ge
  \frac{c_A+c_B}{K}\,\lVert AB\rVert_*^2 .
\]
Here \(K\) is the shared inner dimension, \(c_A,c_B\) are normalized noise
variances set by the two factor grids, and the nuclear norm
\(\lVert AB\rVert_*\) is the sum of the output singular values.  Although the
bound composes two standard inequalities, their composition is tight.  At any
dimension admitting a Hadamard matrix---including every power of two---an
SVD-aligned construction makes the infimum equal to the bound; for every
\(K\), an SVD-aligned DCT construction comes within a factor of two.  If outer rotations are
unavailable, balancing the Gram matrices \(A^\top A\) and \(BB^\top\)
determines the optimal factorization energy, and finite-set flattening reaches
the lower bound within a factor
\(C\log(K(m+n))\).  For power-of-two \(K\), the conditional-expectation method
implements the Hadamard sign-selection stage deterministically in
\(O((m+n)K^2)\) exact-real arithmetic operations.  Synthetic experiments check the Hadamard
full-gauge and inner-only constructions and illustrate the instance-specific tradeoff
between regularization and transform conditioning.  Together, these results
characterize the full-gauge optimum and quantify the inner-only cost of keeping
row and column indices fixed.

\end{abstract}

\section{Introduction}
\label{sec:nuc-intro}

Low-precision matrix multiplication is a central kernel in machine-learning
inference and, increasingly, in scientific computing.  For
\(A\in\R^{m\times K}\) and \(B\in\R^{K\times n}\), a scalar scheme rounds the
factors before accumulating their \(K\)-term inner products.  Here \(K\) is the
shared inner dimension of \(C=AB\).  The quantization error depends
on how \(C\) is factorized, though \(C\) itself does not.  We ask: what is the
smallest leading error achievable over all product-preserving
transformations?

An error in \(A_{ik}\) is multiplied by the row \(B_{k,:}\), and an error in
\(B_{kj}\) is multiplied by the column \(A_{:,k}\).  An entrywise error metric
therefore fails to measure what matters.  Meanwhile, a scalar grid sets its
step from a row or column \emph{range}: the vector's largest absolute entry.
The companion transform study \cite{qmm} develops this
product-weighted error analysis and compares practical transform families.  We
instead seek a lower bound over a broad product-preserving class.

For \(U_L\in\Oo(m)\), \(T\in\GL(K)\), and \(U_R\in\Oo(n)\), consider
\[
 (A,B)\longmapsto (U_LAT,T^{-1}BU_R).
\]
We call the triple \((U_L,T,U_R)\) a \emph{gauge}; a quantity is
\emph{gauge-invariant} if it depends only on \(AB\).  Operationally, the
factors are transformed, \(\ket A=U_LAT\) is quantized row-wise and
\(\ket B=T^{-1}BU_R\) column-wise, the quantized factors are multiplied, and
the outer rotations are undone:
\[
  \that C=U_L^\top(\ket A+E_A)(\ket B+E_B)U_R^\top .
\]

The singular values of the output govern the answer.  Under the model
of \Cref{ass:nuc-model}, no gauge can push the leading expected squared error
below
\begin{equation}
  \Elead\ge \frac{c_A+c_B}{K}\,\nuc{AB}^{\,2}.
  \label{eq:nuc-intro-floor}
\end{equation}
The constants \(c_A,c_B\) are normalized entrywise variances from the two
grid steps; for a signed \(b\)-bit nominal grid,
\(c=1/[12(2^{b-1}-1)^2]\).  The nuclear norm \(\nuc{AB}\) is the sum of the
output singular values.  We call \eqref{eq:nuc-intro-floor} the
\emph{independent-noise floor}: a lower bound invariant under refactorization.

The lower bound follows by composing two standard inequalities: a squared grid
range must cover the average squared energy of the vector it quantizes, and a
Frobenius product must dominate the nuclear norm of its product.  Its primary
value is not the novelty of these inequalities but the tightness of their
composition.  At a \emph{Hadamard dimension}---an order \(K\) admitting a
\(K\times K\) Hadamard matrix, including every power of two---an SVD-aligned
construction makes the infimum equal to \eqref{eq:nuc-intro-floor}.  For every
other dimension, a normalized DCT-II construction remains within a factor of
two.  Thus the floor is both cheap to compute---it reads only the output
singular values---and impossible to undercut.

A two-by-two instance traces the whole argument.  Take
\(A=\operatorname{diag}(4,1)\) and \(B=\operatorname{diag}(1,4)\), so
\(C=\operatorname{diag}(4,4)\) and \(\nuc C=8\).  Quantizing the factors as
given, the leading objective is \(289(c_A+c_B)\), roughly nine times the
floor \(32(c_A+c_B)\).  The diagonal balance \(T=\operatorname{diag}(1/2,2)\)
equalizes both factors to \(\operatorname{diag}(2,2)\) and reduces the leading
objective to \(64(c_A+c_B)\): the factorization energy now matches the nuclear
value, but each row concentrates its energy in one coordinate, so the grid
ranges still exceed the average energy.  Composing with the normalized
\(2\times2\) Hadamard rotation spreads every row and column evenly, and the
objective reaches \(32(c_A+c_B)\): the floor, exactly, with no outer
rotation.  The three stages---floor, balancing, flattening---are the three
sections that follow, and the two loss mechanisms that the example separates
are exactly the two places where the proof of
\eqref{eq:nuc-intro-floor} could be loose.

Over the full gauge, the infimum equals the lower bound at Hadamard dimensions
and is within a universal factor of two for every \(K\).  Outer rotations align
the free indices with the output singular vectors, while an inner orthogonal
matrix spreads each retained singular direction.  A normalized type-II
discrete cosine transform (DCT-II) gives the factor-two bound.  When free
indices carry meaning---tokens, batch items, or output channels---outer
rotations may be unavailable.  In that inner-only setting, Gram-matrix
balancing closes the factorization-energy loss.  We then choose one rotation
after seeing the \(m+n\) rows and columns to quantize, which closes the remaining
peak-to-average gap within a logarithmic factor.  This guarantee applies to the
observed finite set; uniform control over every input is a separate
requirement, and the gap between the two is where the Kashin theory of
redundant representations enters (\Cref{sec:nuc-uncertainty}).

We make five contributions.
\begin{enumerate}
  \item We prove the gauge-invariant lower bound
  \eqref{eq:nuc-intro-floor} under independent, zero-mean per-vector noise,
  and show that its two constituent inequalities are individually sharp.
  \item We prove that the factorization-energy infimum equals
  \(\nuc{AB}^{2}\) for arbitrary factors and give an explicit balancing
  transform for positive-definite Gram matrices.
  \item We show that, over the full gauge, an SVD-aligned construction makes the
  infimum equal to the lower bound at Hadamard dimensions and comes within a
  factor of two for every \(K\); the floor is therefore a characterization,
  not merely a bound.
  \item We show that, for inner-only transforms, finite-set flattening gives an
  \(O(\log(K(m+n)))\) upper factor; for power-of-two \(K\), conditional
  expectations give a deterministic \(O((m+n)K^2)\) exact-real sign-selection
  procedure for the Hadamard flattening stage.
  \item We express the inner-only loss as a peak-to-average comparison and
  connect its uniform, all-input form to Kashin representations.
\end{enumerate}

Recent product-aware methods address complementary quantization settings.
Ordentlich and Polyanskiy analyze direct coding limits for matrix products
\cite{ordentlich2024optimal,ordentlich2026highrate}, while WUSH derives
data-dependent transforms for RTN AbsMax block quantizers from a first-order
stochastic surrogate \cite{wush}.
Practical activation and weight transformations include SmoothQuant, QuaRot,
SpinQuant, and AffineQuant
\cite{xiao2023smoothquant,ashkboos2024quarot,liu2025spinquant,ma2024affinequant}.
These works study coding limits or practical transforms under different
finite-rate constraints; here we identify and characterize a gauge-invariant
floor for product-weighted quantization error.

\Cref{sec:nuc-background} develops the error analysis and notation.
\Cref{sec:nuc-model} states the stochastic model and second-moment identity.
\Cref{sec:nuc-floor,sec:nuc-balancing} prove the lower bound and close its
factorization-energy component.  \Cref{sec:nuc-achievability} gives full-gauge
constant achievability and the inner-only logarithmic construction;
\Cref{sec:nuc-derandomization} derandomizes the latter.
\Cref{sec:nuc-uncertainty} gives the peak-to-average interpretation, and
\Cref{sec:nuc-experiments} evaluates the proved quantities.

\paragraph{Relation to companion work.}
The three manuscripts share the matrix-product setting but answer different
questions; neither companion is a logical premise for the results below.
\Cref{tab:nuc-companion-scope} states the separation.
\begin{table}[H]
\centering
\small
\caption{Separation from companion work.}
\label{tab:nuc-companion-scope}
\begin{tabular}{@{}p{0.22\linewidth}p{0.69\linewidth}@{}}
\toprule
Manuscript & Primary scope\\
\midrule
This paper
& Independent-noise lower bounds over the full gauge, exact or constant
full-gauge achievability, and inner-only logarithmic achievability.\\
Metric companion \cite{metriccompanion}
& Coordinated deterministic choices on fixed grids, formulated as dynamic and
static metric-discrepancy problems with complexity results.\\
Transform companion \cite{qmm}
& Product-weighted error analysis, practical transform families, and empirical
preconditioning comparisons for quantized matrix multiplication.\\
\bottomrule
\end{tabular}
\end{table}

Together, the lower bound and the constructions give a usable certificate for
transform design: the nuclear floor measures unavoidable leading error, while
finite-set flatness measures the additional cost of preserving
application-visible indices.

\section{Background and related work}
\label{sec:nuc-background}

The proof combines four standard ingredients: product-weighted error,
product-preserving transformations, the nuclear-norm factorization identity,
and peak-to-average flattening.  The result---an exact gauge-invariant
characterization at Hadamard dimensions, a constant-factor construction for
every \(K\), and an inner-only logarithmic bound---is new.

\subsection{Product-weighted error}

Write independently perturbed factors as
\(\that A=A+E_A\) and \(\that B=B+E_B\).  Expanding the product gives
\[
  \that A\that B-AB=E_AB+AE_B+E_AE_B.
\]
If the entries of \(E_A\) and \(E_B\) are mutually independent and centered,
then
\begin{equation}
\begin{split}
  \E\,\fro{\that A\that B-AB}^{2}
  ={}&\sum_{i,k}v^A_{ik}\,\tw{B_{k,:}}^{2}
     +\sum_{k,j}v^B_{kj}\,\tw{A_{:,k}}^{2}\\
   &+\sum_k\left(\sum_i v^A_{ik}\right)
            \left(\sum_j v^B_{kj}\right),
\end{split}
\label{eq:nuc-background-master}
\end{equation}
where \(v^A_{ik}\) and \(v^B_{kj}\) are entrywise variances.  The first two
terms show why an entrywise metric is insufficient: the other factor weights
each perturbation.  The final term is the exact bilinear contribution from
perturbing both factors.  The transform companion develops
\eqref{eq:nuc-background-master} for practical transform classes \cite{qmm};
we prove the specialized form in \Cref{sec:nuc-model}.

Subtractive dither adds an independent uniform offset, rounds to the shifted
lattice, and removes the offset.  Without \emph{overload}---leaving the
representable interval---the error is uniform over one grid cell and independent
of the input \cite{graystockham}.  Treating deterministic round-to-nearest error
as locally uniform is instead a high-resolution approximation that requires
negligible clipping \cite{grayneuhoff}.  The variances in
\eqref{eq:nuc-background-master} depend on how the factors are represented; the
next subsection specifies the allowed transformations.

\subsection{Gauge invariance}

For every \(T\in\GL(K)\), the inner change of basis
\[
  AB=(AT)(T^{-1}B)
\]
preserves the product.  Orthogonal transformations on the free indices extend
this freedom: for \(U_L\in\Oo(m)\) and \(U_R\in\Oo(n)\),
\begin{equation}
 U_L^\top\big[(U_LAT)(T^{-1}BU_R)\big]U_R^\top=AB.
 \label{eq:nuc-gauge-invariance}
\end{equation}
The full gauge is \(\Oo(m)\times\GL(K)\times\Oo(n)\).  Inner transforms reshape
both factors; each outer transform changes the ranges of its own factor.  Outer
rotations matter: \Cref{thm:nuc-full-gauge} shows that they close the logarithmic
gap left by inner transforms alone.

The infimum of \(\fro{AT}\fro{T^{-1}B}\) over the inner orbit is governed by one
output invariant.

\subsection{Nuclear-norm factorization}

The nuclear norm \(\nuc C=\sum_\ell\sigma_\ell(C)\) is orthogonally invariant
and satisfies
\begin{equation}
  \nuc C=\min_{XY=C}\fro X\,\fro Y
  \label{eq:nuc-factorization}
\end{equation}
when the inner dimension is at least \(\rank(C)\)
\cite{srebro2005rank}.  Equivalently,
\(\nuc{XY}\le \fro X\,\fro Y\).  If \(\rank(C)=K\), any two
inner-dimension-\(K\) factorizations of \(C\) differ by an invertible change of
basis.  This single-orbit statement fails when \(\rank(C)<K\), but regularized
gauges can squeeze directions that do not contribute to the product.
\Cref{lem:nuc-balance} proves that the orbit infimum remains \(\nuc{AB}\).

The Gram-matrix equalization used later parallels classical balanced state-space
realizations \cite{moore1981balanced}.  Here the matrices arise from a
factorized product, and we optimize a multiplicative trace objective.

\subsection{Range versus energy}

A scalar grid sets its step from the largest absolute coordinate of the vector
it quantizes, namely \(\ii z\).  A single large coordinate therefore stretches
the grid and coarsens the resolution available to every other entry.  The
lower bound, by contrast, depends on the vector's average energy, namely
\(\tw z^2/K\).  For \(z\in\R^K\),
\begin{equation}
  \frac{\tw z^2}{K}\le \ii z^2\le \tw z^2.
  \label{eq:nuc-range-energy}
\end{equation}
The ratio
\(K\ii z^2/\tw z^2\), which we call the \emph{peak-to-average ratio}, measures
this gap.  It equals one when the coordinates have equal magnitude and can
reach \(K\) for a coordinate spike.  Minimizing this ratio---flattening the
vector so that no coordinate dominates---is therefore the key to reaching the
lower bound.  A shared rotation seeks to flatten all rows of one factor and all
columns of the other.  Over this finite set, random rotations lose only a
logarithmic factor; uniform control of every input instead requires the
redundant representations discussed in \Cref{sec:nuc-uncertainty}.

\subsection{Relation to quantization methods}

Ordentlich and Polyanskiy study direct coding bounds for matrix products
\cite{ordentlich2024optimal} and high-rate rate--distortion comparisons for
scalar INT, floating-point, and lattice schemes
\cite{ordentlich2026highrate}.  Their arbitrary-rate coding benchmark and
analysis of practical AbsMax INT, floating-point, and lattice schemes differ
from our fixed-pair scalar row/column-grid lower bound.  Their Part II
treats the distinct covariance-informed, weight-only setting
\cite{ordentlich2026highrate2}.

WUSH targets RTN
AbsMax-scaled block quantizers and optimizes a distributional, first-order
stochastic surrogate built from blockwise second moments \cite{wush}; our
finite-set construction instead sees one fixed collection of rows and columns.
The nuclear norm adds a factorization-independent lower bound rather than
another transform recipe.

Practical methods usually omit outer rotations because free indices can encode
tokens, batch elements, or output channels.  SmoothQuant, QuaRot, SpinQuant,
and AffineQuant
\cite{xiao2023smoothquant,ashkboos2024quarot,liu2025spinquant,ma2024affinequant}
therefore operate under different structural and finite-rate constraints.  The
full-gauge theorem characterizes the mathematical transformation class, while
the inner-only theorem is the closer comparison when free indices must remain
fixed.

\section{Noise model and second-moment identity}
\label{sec:nuc-model}

The nuclear lower bound requires fixing the stochastic assumptions and the
\(m\) row ranges and \(n\) column ranges.  We define the transformed factors,
derive their normalized noise constants, and separate the leading objective
from the bilinear correction.

\subsection{Transformed factors and scalar grids}

For a gauge
\((U_L,T,U_R)\in\Oo(m)\times\GL(K)\times\Oo(n)\), define
\begin{equation}
  \ket A=U_LAT,\qquad \ket B=T^{-1}BU_R.
  \label{eq:nuc-transformed-factors}
\end{equation}
A nominal signed \(b\)-bit symmetric grid with range \(r\) has
\(q=2^{b-1}-1\) positive steps, levels
\(\{jr/q:j=-q,\ldots,q\}\), and cell width \(\Delta=r/q\).
Each row of \(\ket A\) and each column of \(\ket B\) selects its own range:
\[
  r_i=\ii{\ket A_{i,:}},\qquad
  \gamma_j=\ii{\ket B_{:,j}}.
\]
For bit widths \(b_A,b_B\ge2\), set
\begin{equation}
  c_A=\frac{1}{12(2^{b_A-1}-1)^2},\qquad
  c_B=\frac{1}{12(2^{b_B-1}-1)^2}.
  \label{eq:nuc-quantization-constants}
\end{equation}
The factor \(1/12\) is the variance of a uniform variable on a unit-width
interval: subtractive dither on a cell of width \(\Delta=r/q\) has variance
\(\Delta^2/12\).  The nominal bit width is one way to set the step; the exact
model below depends only on the resulting step.

\begin{assumption}[Dithered per-vector model]
\label{ass:nuc-model}
All entries of \(E_A\) and \(E_B\) are mutually independent and have zero mean,
with
\begin{equation}
  \Var(E^A_{ik})=c_A r_i^2,\qquad
  \Var(E^B_{kj})=c_B \gamma_j^2.
  \label{eq:nuc-variance-fields}
\end{equation}
Entrywise independence within each row and column is part of the assumption:
\emph{per-vector} refers to the shared range, not to shared vector-level dither.
\end{assumption}

\begin{center}
\small
\begin{tabular}{@{}p{0.18\linewidth}p{0.74\linewidth}@{}}
\textbf{Exact theorem:} & independent subtractive dither on an unbounded
lattice, hence no overload;\\
\textbf{Approximation:} & the high-resolution granular model for a finite grid
with a guard band, an interior margin that makes clipping negligible;\\
\textbf{Excluded:} & deterministic round-to-nearest with clipping or
correlated errors.
\end{tabular}
\end{center}

Define
\[
 \that C
 =U_L^\top(\ket A+E_A)(\ket B+E_B)U_R^\top .
\]
Here, \(\that C\) is the product of the quantized transformed factors after
undoing the outer rotations.  Orthogonal invariance of the Frobenius norm lets
us analyze the transformed product directly.

\subsection{Per-vector second-moment identity}

To understand how a gauge affects the final product error, we first separate
the exact expectation into its constituent parts.  The identity below is the
computational core of the paper: it splits the error into two leading
contributions driven by the grid ranges and a higher-order bilinear correction.
Each leading contribution multiplies the total squared range on one side by
the full energy of the opposite factor.  This multiplicative structure
explains why the lower bound in \Cref{sec:nuc-floor} naturally contains a
product of factor energies, which is then bounded by an output-only quantity.

\begin{lemma}[Second-moment identity; per-vector form]
\label{lem:nuc-master}
Under \Cref{ass:nuc-model},
\begin{equation}
\begin{split}
 \E\,\fro{\that C-C}^{2}
 ={}&
 \underbrace{c_A\fro{\ket B}^{2}\sum_i r_i^2
 +c_B\fro{\ket A}^{2}\sum_j\gamma_j^2}_{\Elead}\\
 &+
 \underbrace{Kc_Ac_B
 \left(\sum_i r_i^2\right)\left(\sum_j\gamma_j^2\right)}
 _{E_{\mathrm{bil}}}.
\end{split}
\label{eq:nuc-leading-identity}
\end{equation}
\end{lemma}

\begin{proof}
Orthogonal invariance gives
\[
 \fro{\that C-C}
 =\fro{E_A\ket B+\ket AE_B+E_AE_B}.
\]
Expand the square on the right.  Every mixed inner product has expectation zero
because each summand contains a
lone centered error independent of the remaining factors.  For the first
square,
\[
 \E\,\fro{E_A\ket B}^{2}
 =\sum_{i,j}\sum_{k,k'}
   \E[E^A_{ik}E^A_{ik'}]\ket B_{kj}\ket B_{k'j}.
\]
Only \(k=k'\) remains, so the first two squares are
\[
 \E\,\fro{E_A\ket B}^{2}
 =c_A\Big(\sum_i r_i^2\Big)\fro{\ket B}^{2},
 \qquad
 \E\,\fro{\ket AE_B}^{2}
 =c_B\Big(\sum_j\gamma_j^2\Big)\fro{\ket A}^{2}.
\]
In the bilinear square, independence between the two error matrices and within
their entries again leaves only equal contracted indices:
\[
 \E\,\fro{E_AE_B}^{2}
 =Kc_Ac_B\left(\sum_i r_i^2\right)
             \left(\sum_j\gamma_j^2\right).
\]
Adding the three squares proves \eqref{eq:nuc-leading-identity}.
\end{proof}

\subsection{High-rate ordering}

The term \(\Elead\) is first order in the common variance scale (quadratic in
grid width), whereas the bilinear term is second order.  Write
\(c_A=\eta\bar c_A\), \(c_B=\eta\bar c_B\), and, for a fixed gauge \(G\),
\begin{equation}
  E_\eta(G)=\eta\,\Lambda(G)+\eta^2\mathcal R(G),
  \qquad \mathcal R(G)\ge0,
  \label{eq:nuc-high-rate-functionals}
\end{equation}
where \(\Lambda\) is the leading
objective of \eqref{eq:nuc-leading-identity} at unit noise scale.  This
high-rate ordering need not be uniform over
\(\eta\)-dependent ill-conditioned gauges, so one cannot simply discard
\(\mathcal R\) before taking an infimum: an \(\eta\)-dependent gauge could make
the nominally higher-order contribution dominate.

The required asymptotic nevertheless follows from a sandwich.  Let
\(\lambda_*=\inf_G\Lambda(G)\).  Nonnegativity gives
\[
  \inf_G E_\eta(G)\ge\eta\lambda_*.
\]
For any \(\rho>0\), choose one fixed gauge \(G_\rho\) with
\(\Lambda(G_\rho)\le\lambda_*+\rho\).  Then
\[
 \inf_G E_\eta(G)
 \le E_\eta(G_\rho)
 \le \eta(\lambda_*+\rho)+\eta^2\mathcal R(G_\rho).
\]
Divide by \(\eta\), take \(\eta\downarrow0\), and then take
\(\rho\downarrow0\).  The two bounds give
\begin{equation}
  \inf_G E_\eta(G)=\eta\inf_G\Lambda(G)+o(\eta).
  \label{eq:nuc-high-rate-sandwich}
\end{equation}
Thus \(\Elead\) is the correct transform objective to first order, while
\(E_{\mathrm{bil}}\) remains part of the exact expected error at fixed grid
steps.
The limit in \eqref{eq:nuc-high-rate-sandwich} identifies the sharp leading
objective, while \Cref{cor:nuc-full-error-floor} gives a nonasymptotic lower
bound for the complete error.

\Cref{tab:nuc-notation} collects the symbols used by the remaining arguments.
\begin{table}[t]
\centering
\small
\caption{Core notation.}
\label{tab:nuc-notation}
\begin{tabular}{@{}ll@{}}
\toprule
Symbol & Meaning\\
\midrule
\(A,B,C=AB\) & factors and output in
\(\R^{m\times K},\R^{K\times n},\R^{m\times n}\)\\
\(i,j;k\) & free output indices; contracted inner index\\
\((U_L,T,U_R)\) & gauge in \(\Oo(m)\times\GL(K)\times\Oo(n)\)\\
\(r_i,\gamma_j\) & row and column ranges of transformed factors\\
\(c_A,c_B\) & normalized scalar-noise variances\\
\(P,Q\) & inner Gram matrices \(A^\top A\) and \(BB^\top\)\\
\(\Elead\) & leading expected squared product error\\
\bottomrule
\end{tabular}
\end{table}

\section{A gauge-invariant nuclear lower bound}
\label{sec:nuc-floor}

The second-moment identity reduces preprocessing to an optimization over the
ranges \(\{r_i\}\) and \(\{\gamma_j\}\).  Although these ranges depend on the
factorization, the leading error cannot fall below an output-only threshold.
The proof uses two individually standard inequalities; over the full gauge,
their composition is exact at Hadamard dimensions and within a factor of two
for every \(K\).

\begin{theorem}[Gauge-invariant nuclear lower bound]
\label{thm:nuc-floor}
Let \(A\in\R^{m\times K}\) and \(B\in\R^{K\times n}\).  Under
\Cref{ass:nuc-model}, every
\((U_L,T,U_R)\in\Oo(m)\times\GL(K)\times\Oo(n)\) satisfies
\begin{equation}
  \Elead(U_L,T,U_R)
  \ge \frac{c_A+c_B}{K}\,\nuc{AB}^{\,2}.
  \label{eq:nuc-floor}
\end{equation}
\end{theorem}

\begin{proof}
The bound composes two inequalities.  The first is elementary: for each row
of \(\ket A\), the range--energy inequality
\eqref{eq:nuc-range-energy} gives
\[
  \ii{\ket A_{i,:}}^2\ge
  \frac{1}{K}\tw{\ket A_{i,:}}^2.
\]
Summing over the rows of \(\ket A\) and columns of \(\ket B\) yields
\[
  \sum_i r_i^2\ge\frac1K\fro{\ket A}^2,
  \qquad
  \sum_j\gamma_j^2\ge\frac1K\fro{\ket B}^2.
\]
Substituting these inequalities into the leading part of
\eqref{eq:nuc-leading-identity} yields
\[
  \Elead
  \ge\frac{c_A+c_B}{K}
      \fro{\ket A}^2\fro{\ket B}^2.
\]
The second inequality is the Frobenius--nuclear product bound
\eqref{eq:nuc-factorization}; with orthogonal invariance it gives
\[
  \fro{\ket A}\fro{\ket B}
  \ge\nuc{\ket A\ket B}
  =\nuc{U_LABU_R}
  =\nuc{AB}.
\]
Squaring this inequality proves \eqref{eq:nuc-floor}.
\end{proof}

Neither step uses the gauge beyond orthogonal invariance, and each is
individually sharp: a row whose coordinates have equal magnitude makes the
first inequality an equality, and a factorization with
\(\fro X\fro Y=\nuc{XY}\) makes the second one an equality.  The bound is
therefore not an artifact of loose estimation.  Simultaneous tightness is
addressed by \Cref{thm:nuc-full-gauge}, which proves equality at Hadamard
dimensions and a factor-two match for every \(K\).

\begin{corollary}[Complete-error lower bound]
\label{cor:nuc-full-error-floor}
Under \Cref{ass:nuc-model}, every gauge satisfies
\begin{equation}
 \E\,\fro{\that C-C}^{2}
 \ge
 \frac{c_A+c_B+c_Ac_B}{K}\,\nuc{AB}^{\,2}.
 \label{eq:nuc-full-error-floor}
\end{equation}
This nonasymptotic bound applies to the exact expected error.
\end{corollary}

\begin{proof}
The range--energy bounds from the theorem give
\[
 E_{\mathrm{bil}}
 =Kc_Ac_B\left(\sum_i r_i^2\right)
             \left(\sum_j\gamma_j^2\right)
 \ge \frac{c_Ac_B}{K}\fro{\ket A}^{2}\fro{\ket B}^{2}
 \ge \frac{c_Ac_B}{K}\nuc{AB}^{\,2}.
\]
Adding this inequality to \eqref{eq:nuc-floor} proves
\eqref{eq:nuc-full-error-floor}.  At Hadamard dimensions, the construction in
\Cref{thm:nuc-full-gauge} makes both bounds tight when
\(\operatorname{rank}(AB)=K\) and asymptotically tight when
\(\operatorname{rank}(AB)<K\).  In the latter case, the asymptotic qualifier
comes only from the limiting factorization balance: the retained Hadamard rows
and columns are already exactly flat.
\end{proof}

The proof leaves two gaps: factorization energy may exceed the output nuclear
norm, and squared range may exceed average energy.  The next sections close the
factorization-energy gap and control the range--energy gap.

\section{Balancing the Gram matrices}
\label{sec:nuc-balancing}

Our first step toward the lower bound is to remove directional imbalance from
the factorization.  One factor may load a latent direction heavily while the
other loads it lightly; a compensating inner transformation preserves the
product but can reduce the resulting factorization-energy cost.  Let
\(P=A^\top A\) and \(Q=BB^\top\) be the inner Gram matrices, which record how
the two factors distribute energy along the contracted dimension.  Balancing
them makes the factors share the magnitudes of the latent features and
minimizes the factorization-energy loss.  This leaves one source of excess
error: coordinate spikes that convert energy into unnecessarily wide grids.

\begin{lemma}[Gram-matrix balancing]
\label{lem:nuc-balance}
For arbitrary \(A\in\R^{m\times K}\) and \(B\in\R^{K\times n}\),
\begin{equation}
 \inf_{T\in\GL(K)}
 \fro{AT}^{2}\fro{T^{-1}B}^{2}
 =\nuc{AB}^{\,2}.
 \label{eq:nuc-balance-value}
\end{equation}
Equivalently,
\(
 \inf_{T\in\GL(K)}\fro{AT}\fro{T^{-1}B}=\nuc{AB}.
\)
If \(P,Q\succ0\), the infimum is attained by any \(T\) satisfying
\begin{equation}
 TT^\top
 =\beta P^{-1/2}
  (P^{1/2}QP^{1/2})^{1/2}
  P^{-1/2},
 \qquad \beta>0.
 \label{eq:nuc-balancing-transform}
\end{equation}
If either Gram matrix is singular, equality still holds as an infimum but may not
be attained at finite condition number.
\end{lemma}

\begin{proof}
First suppose \(P,Q\succ0\) and set \(S=TT^\top\succ0\).  Cyclicity of the
trace gives
\[
  \fro{AT}^{2}=\tr(PS),\qquad
  \fro{T^{-1}B}^{2}=\tr(QS^{-1}).
\]
With \(R=P^{1/2}SP^{1/2}\) and
\(M=P^{1/2}QP^{1/2}\), the objective becomes
\(\tr(R)\tr(MR^{-1})\).  Trace Cauchy--Schwarz applies without a commutativity
assumption:
\[
 \tr(M^{1/2})
 =\tr\!\left(R^{1/2}R^{-1/2}M^{1/2}\right)
 \le \tr(R)^{1/2}\tr(MR^{-1})^{1/2}.
\]
Equality holds precisely when
    \(R=\beta M^{1/2}\), which reverses to
\eqref{eq:nuc-balancing-transform}.  The nonzero eigenvalues of \(M\) are the
squared singular values of \(AB\), so
\(\tr(M^{1/2})=\nuc{AB}\).  This proves both the value and attainment in the
positive-definite case.

For singular Gram matrices, define
\[
  P_\varepsilon=P+\varepsilon I,\qquad
  Q_\varepsilon=Q+\varepsilon I,
\]
and let \(S_\varepsilon\) attain the regularized positive-definite problem.  Its
value is
\[
 v_\varepsilon=
 \left[
 \tr\!\left(
 (P_\varepsilon^{1/2}Q_\varepsilon
 P_\varepsilon^{1/2})^{1/2}
 \right)
 \right]^2.
\]
For every \(S\succ0\), the Frobenius--nuclear product inequality gives
\(\tr(PS)\tr(QS^{-1})\ge\nuc{AB}^2\).  Conversely,
\(P\preceq P_\varepsilon\) and \(Q\preceq Q_\varepsilon\) imply
\[
 \inf_{S\succ0}\tr(PS)\tr(QS^{-1})
 \le \tr(PS_\varepsilon)\tr(QS_\varepsilon^{-1})
 \le v_\varepsilon.
\]
Continuity of the positive-semidefinite square root gives
\[
 v_\varepsilon\longrightarrow
 \left[\tr\!\left(
 (P^{1/2}QP^{1/2})^{1/2}
 \right)\right]^2
 =\nuc{AB}^{\,2}.
\]
The lower and upper bounds squeeze the singular infimum to the same value.
The associated \(S_\varepsilon\), and hence \(T_\varepsilon\), may become
ill-conditioned, which explains the possible failure of attainment.
\end{proof}

Balancing solves the factorization-energy objective but not the scalar range
objective.  A row can still concentrate most of its energy in one coordinate,
inflating the squared grid step by as much as \(K\) (the step itself by
\(\sqrt K\)).  In
the two-by-two example of \Cref{sec:nuc-intro}, the balance
\(T=\operatorname{diag}(1/2,2)\) attains the infimum
\eqref{eq:nuc-balance-value}---the balanced factors satisfy
\(\fro{\ket A}\fro{\ket B}=8=\nuc{AB}\)---yet the leading objective remains
twice the floor, entirely because the rows of \(\ket A\) and columns of
\(\ket B\) are coordinate spikes.
The next section first closes this loss over the full gauge and then studies
the inner-only restriction.

\section{Full-gauge and inner-only achievability}
\label{sec:nuc-achievability}

The lower bound becomes an exact or constant-factor characterization once outer
rotations are allowed.  When those rotations are unavailable, Gram-matrix
balancing followed by finite-set flattening gives a logarithmic inner-only
bound.  We state the two regimes separately.

\subsection{Constant achievability over the full gauge}

For \(H\in\Oo(K)\), write
\[
 \lVert H\rVert_{\max}=\max_{a,b}|H_{ab}|,
 \qquad L_H=K\lVert H\rVert_{\max}^{2}.
\]

\begin{theorem}[Full-gauge constant achievability]
\label{thm:nuc-full-gauge}
Let \(C=AB\ne0\) have rank \(r\) and thin SVD
\(C=U\Sigma V^\top\).  For every \(H\in\Oo(K)\),
\begin{equation}
 \inf_{\substack{U_L\in\Oo(m),\,T\in\GL(K),\\U_R\in\Oo(n)}}
 \Elead(U_L,T,U_R)
 \le
 L_H\,\frac{c_A+c_B}{K}\nuc C^{\,2}.
 \label{eq:nuc-full-gauge-bound}
\end{equation}
If \(r=K\), the bound is witnessed by a single finite gauge.  If \(r<K\), the
same bound holds as an infimum and finite-condition-number attainment may fail.
Consequently, a normalized Hadamard matrix makes the infimum equal to
\eqref{eq:nuc-floor} whenever one of order \(K\) exists, with finite attainment
guaranteed when \(r=K\); the normalized DCT-II gives a factor at most two for
every \(K\).
\end{theorem}

\begin{proof}
First suppose \(r=K\).  Then
\(\operatorname{range}(A)=\operatorname{range}(U)\), so
\(A=UN\) for \(N=U^\top A\in\GL(K)\), and
\(B=N^{-1}\Sigma V^\top\).  Set
\(T_0=N^{-1}\Sigma^{1/2}\).  Choose orthogonal \(U_L,U_R\) by extending
\(U,V\) to bases so that
\[
 U_LU=\begin{bmatrix}I_K\\0\end{bmatrix},
 \qquad
 V^\top U_R=\begin{bmatrix}I_K&0\end{bmatrix}.
\]
With \(T=T_0H\), the transformed factors are
\[
 \ket A=
 \begin{bmatrix}\Sigma^{1/2}H\\0\end{bmatrix},
 \qquad
 \ket B=
 \begin{bmatrix}H^\top\Sigma^{1/2}&0\end{bmatrix}.
\]
Therefore
\[
 \sum_i r_i^2\le\lVert H\rVert_{\max}^{2}\tr\Sigma,\qquad
 \sum_j\gamma_j^2\le\lVert H\rVert_{\max}^{2}\tr\Sigma,
 \qquad
 \fro{\ket A}^{2}=\fro{\ket B}^{2}=\tr\Sigma=\nuc C.
\]
Substitution into \eqref{eq:nuc-leading-identity} proves
\eqref{eq:nuc-full-gauge-bound}.  A normalized Hadamard matrix has
\(\lVert H\rVert_{\max}^{2}=1/K\), so both range--energy inequalities are
equalities.  A normalized DCT-II has
\(\lVert H\rVert_{\max}^{2}\le2/K\).

Now let \(r<K\).  The same geometry appears only as a limit: an
energy-minimizing sequence converges to an SVD-aligned factorization, after
which \(H\) spreads the retained singular directions.  Take a sequence \(T_\ell\) from
\Cref{lem:nuc-balance} and rescale each \(T_\ell\) by a positive scalar so that
\[
 \fro{AT_\ell}^{2}=\fro{T_\ell^{-1}B}^{2}\longrightarrow\nuc C.
\]
The two factor sequences are bounded.  Along a subsequence they converge to
\(X,Y\) with \(XY=C\) and
\(\fro X^{2}=\fro Y^{2}=\nuc C=:s\).  Moreover,
\[
 s=\tr(U^\top XYV)
 \le \fro{U^\top X}\fro{YV}
 \le \fro X\fro Y=s.
\]
Equality in the projection bounds gives \(X=UX_0\) and \(Y=Y_0V^\top\).
Equality in trace Cauchy--Schwarz, together with the equal norms, gives
\(Y_0=X_0^\top\), while \(XY=C\) gives \(X_0X_0^\top=\Sigma\).  Hence
\[
 X=U\Sigma^{1/2}R^\top,\qquad
 Y=R\Sigma^{1/2}V^\top
\]
for some \(R\in\R^{K\times r}\) with \(R^\top R=I_r\).  Complete \(R\) to
\(\bar R=[R,R_\perp]\in\Oo(K)\), and choose \(U_L,U_R\) from orthogonal
extensions of \(U,V\).  The gauges
\((U_L,T_\ell\bar R H,U_R)\) converge in transformed-factor space to
\[
 \begin{bmatrix}\Sigma^{1/2}[I_r\ 0]H\\0\end{bmatrix},
 \qquad
 \begin{bmatrix}H^\top[I_r\ 0]^\top\Sigma^{1/2}&0\end{bmatrix}.
\]
When \(H\) is Hadamard, \([I_r\ 0]H\) selects \(r\) complete Hadamard rows.
The retained rows and columns are therefore still exactly flat; only the
factorization-energy balance requires a limiting sequence.
Continuity of the full leading expression in the transformed factors gives the
same upper bound as an infimum.  Combining the Hadamard case with
\Cref{thm:nuc-floor} proves equality of the infimum with the lower bound.
\end{proof}

\begin{remark}[Complete error under the same construction]
\label{rem:nuc-full-gauge-complete}
The construction with flatness factor \(L_H\) also gives
\[
 \E\,\fro{\that C-C}^{2}
 \le
 \frac{L_H(c_A+c_B)+L_H^2c_Ac_B}{K}\nuc C^{\,2}
 +o_\ell(1)
\]
in singular cases.  Thus the bilinear upper contribution carries the square of
the flattening factor.  Hadamard gauges attain or approach the complete-error
lower bound \eqref{eq:nuc-full-error-floor}; the DCT factor-two statement applies
to the leading term, not to the complete error.
\end{remark}

\subsection{Logarithmic achievability by inner transforms}

We now restrict to \(U_L=I\) and \(U_R=I\), as required when free indices carry
meaning and may not be rotated.  For \(T\in\GL(K)\), define
\begin{equation}
\begin{split}
 \Phi_\infty(T)={}&
 c_A\fro{T^{-1}B}^{2}
      \sum_i\ii{(AT)_{i,:}}^2\\
 &+c_B\fro{AT}^{2}
      \sum_j\ii{(T^{-1}B)_{:,j}}^2.
\end{split}
\label{eq:nuc-range-objective}
\end{equation}
Thus \(\Phi_\infty(T)=\Elead(I,T,I)\); we use \(\Phi_\infty\) when only
the inner transform varies.

\begin{theorem}[Inner-only bound from finite-set flattening]
\label{thm:nuc-achievability}
For each \(\varepsilon>0\), set
\[
 P_\varepsilon=A^\top A+\varepsilon I,\qquad
 Q_\varepsilon=BB^\top+\varepsilon I,
\]
and choose \(T_\varepsilon\) with
\begin{equation}
\begin{split}
 T_\varepsilon T_\varepsilon^\top
 ={}&P_\varepsilon^{-1/2}
 (P_\varepsilon^{1/2}Q_\varepsilon
  P_\varepsilon^{1/2})^{1/2}
 P_\varepsilon^{-1/2}.
\end{split}
\label{eq:nuc-regularized-transform}
\end{equation}
Let \(x_i^{(\varepsilon)}\) be the rows of \(AT_\varepsilon\) and
\(y_j^{(\varepsilon)}\) the columns of \(T_\varepsilon^{-1}B\).  Suppose a
family \(W_\varepsilon\in\Oo(K)\) satisfies

\begin{equation}
\begin{split}
 \ii{x_i^{(\varepsilon)}W_\varepsilon}^{2}
 &\le \frac{L}{K}\tw{x_i^{(\varepsilon)}}^{2},\\
 \ii{W_\varepsilon^\top y_j^{(\varepsilon)}}^{2}
 &\le \frac{L}{K}\tw{y_j^{(\varepsilon)}}^{2}
\end{split}
\label{eq:nuc-flattening-hypothesis}
\end{equation}
with a single numerical constant \(L\) independent of \(\varepsilon\).  Then,
as \(\varepsilon\downarrow0\),
\begin{equation}
 \Elead(I,T_\varepsilon W_\varepsilon,I)
 \le
 L\,\frac{c_A+c_B}{K}\nuc{AB}^{\,2}
 +o_\varepsilon(1).
 \label{eq:nuc-achievability-bound}
\end{equation}
\end{theorem}

\begin{proof}
Set \(X_\varepsilon=AT_\varepsilon\) and
\(Y_\varepsilon=T_\varepsilon^{-1}B\).  Under
\eqref{eq:nuc-flattening-hypothesis}, summing over the relevant rows and columns
gives
\[
 \sum_i\ii{(X_\varepsilon W_\varepsilon)_{i,:}}^2
 \le\frac{L}{K}\fro{X_\varepsilon}^{2},\qquad
 \sum_j\ii{(W_\varepsilon^\top Y_\varepsilon)_{:,j}}^2
 \le\frac{L}{K}\fro{Y_\varepsilon}^{2}.
\]
Orthogonality preserves Frobenius norms, so
\[
 \Elead
 \le L\frac{c_A+c_B}{K}
 \fro{X_\varepsilon}^{2}\fro{Y_\varepsilon}^{2}.
\]
\Cref{lem:nuc-balance} gives
\(\fro{X_\varepsilon}^{2}\fro{Y_\varepsilon}^{2}
  \to\nuc{AB}^{2}\), proving \eqref{eq:nuc-achievability-bound}.
\end{proof}

\begin{corollary}[Log-achievability by inner transforms]
\label{cor:nuc-log-achievability}
Under the notation of \Cref{thm:nuc-achievability}, there is an absolute
constant \(C_0>0\) such that, for fixed \(\varepsilon>0\), a Haar-random
\(W_\varepsilon\) satisfies \eqref{eq:nuc-flattening-hypothesis} with
\[
 L=C_0\log\!\left(\frac{K(m+n)}{\delta}\right)
\]
with probability at least \(1-\delta\) for every \(\delta\in(0,1)\).  For
power-of-two \(K\), randomized
Hadamard signs give the same order, while
\Cref{thm:nuc-derandomization} gives the explicit deterministic value
\(L=2\log(4K(m+n))\).  Consequently, for an absolute constant \(C>0\),
\begin{equation}
 \frac{c_A+c_B}{K}\nuc{AB}^{\,2}
 \le \inf_{T\in\GL(K)}\Phi_\infty(T)
 \le
 C\frac{(c_A+c_B)\log(K(m+n))}{K}\nuc{AB}^{\,2}.
\label{eq:nuc-two-sided}
\end{equation}
\end{corollary}

\begin{proof}
For the probability statement, fix \(\varepsilon>0\), set
\(X_\varepsilon=AT_\varepsilon\) and
\(Y_\varepsilon=T_\varepsilon^{-1}B\), and write \(N=m+n\).  Let
\(v\) be the transpose of any row of \(X_\varepsilon\) or any column of
\(Y_\varepsilon\).  Under a Haar rotation,
\(W_\varepsilon^\top v\) is uniform on the Euclidean sphere of radius
\(\tw v\), and spherical concentration gives
\[
 \mathbb P\{|(W_\varepsilon^\top v)_k|>t\tw v\}
 \le2e^{-cKt^2}.
\]
A union bound over \(K\) coordinates and the \(N\) rows and columns gives
\(L=C_0\log(KN/\delta)\).  This expression depends only on \(K,N,\delta\), not
on \(\varepsilon\), although a fresh \(W_\varepsilon\) may be selected at each
regularization level.  For power-of-two \(K\), Hoeffding's inequality gives the
same order for randomized Hadamard signs, and
\Cref{thm:nuc-derandomization} gives the stated deterministic value of \(L\).
Choose such a \(W_\varepsilon\) along a sequence
\(\varepsilon\downarrow0\).  Each \(T_\varepsilon W_\varepsilon\) is feasible
for the inner-only infimum, so \Cref{thm:nuc-achievability} and passage to the
limit remove the \(o_\varepsilon(1)\) term in the upper bound.  The lower bound
is \Cref{thm:nuc-floor} with \(U_L=U_R=I\), proving
\eqref{eq:nuc-two-sided}.
\end{proof}

No outer rotation is needed for this bound because
\(T_\varepsilon W_\varepsilon\in\GL(K)\);
\Cref{thm:nuc-full-gauge} shows exactly what the outer rotations add.
In the two-by-two example of \Cref{sec:nuc-intro}, the composition
\(T H\) with the normalized Hadamard \(H\) is exactly the construction of
\Cref{thm:nuc-full-gauge} with the outer rotations set to the identity.  The
output \(4I\) is already SVD-aligned, and the objective reaches the floor
\(32(c_A+c_B)\) with no rotation of the free indices.

\section{Deterministic finite-set flattening}
\label{sec:nuc-derandomization}

The randomized flattening of \Cref{cor:nuc-log-achievability} can be verified after
sampling, but some settings require a deterministic certificate for the
realized instance.  For a power-of-two dimension, we use conditional
expectations to choose Hadamard signs without weakening the logarithmic
guarantee.

\begin{theorem}[Finite-set derandomization]
\label{thm:nuc-derandomization}
Let \(x_1,\ldots,x_N\in\R^d\), where \(d\) is a power of two, and let
\(H\in\R^{d\times d}\) be a normalized Hadamard matrix.  One can compute
\(s\in\{\pm1\}^d\) deterministically in \(O(Nd^2)\) exact-real arithmetic
operations such that every \(j\)
satisfies
\begin{equation}
 \ii{H\diag(s)x_j}
 \le\tw{x_j}\sqrt{\frac{2\log(4Nd)}{d}}.
 \label{eq:nuc-derandomization-bound}
\end{equation}
\end{theorem}

\begin{proof}
Discard zero vectors and normalize \(u_j=x_j/\tw{x_j}\).  Let
\(\sigma_1,\ldots,\sigma_d\) be independent Rademacher signs and define
\[
 Z_{jk}=(H\diag(\sigma)u_j)_k
 =\sum_{\ell=1}^d H_{k\ell}\sigma_\ell u_{j\ell}.
\]
Because \(H_{k\ell}^2=1/d\),
\(\sum_\ell H_{k\ell}^2u_{j\ell}^2=1/d\).
For
\(\tau=\sqrt{2\log(4Nd)/d}\), Hoeffding's inequality gives
\[
  \Pr\{|Z_{jk}|>\tau\}
  \le2e^{-d\tau^2/2}
  =\frac{1}{2Nd}.
\]
A union bound shows that a successful sign vector exists.

To find one, fix \(s_1,\ldots,s_t\) and write
\[
 Z_{jk}=a_{jk}^{(t)}+R_{jk}^{(t)},\qquad
 a_{jk}^{(t)}=\sum_{\ell\le t}H_{k\ell}s_\ell u_{j\ell}.
\]
We choose the remaining signs by the method of conditional probabilities,
implemented through conditional-expectation choices with an exponential-moment
pessimistic estimator \cite{raghavan1988probabilistic}.
For \(\lambda>0\), use the pessimistic estimator
\begin{equation}
\begin{split}
 \Psi_t=e^{-\lambda\tau}\sum_{j,k}\bigg[
 &e^{\lambda a_{jk}^{(t)}}
  \prod_{\ell>t}\cosh(\lambda H_{k\ell}u_{j\ell})\\
 +{}&
 e^{-\lambda a_{jk}^{(t)}}
  \prod_{\ell>t}\cosh(\lambda H_{k\ell}u_{j\ell})
 \bigg].
\end{split}
\label{eq:nuc-pessimistic-estimator}
\end{equation}
At a complete assignment, any violation \(|Z_{jk}|>\tau\) makes one exponential
larger than one, so \(\Psi_d<1\) certifies success.  At \(t=0\), the inequality
\(\cosh z\le e^{z^2/2}\) gives
\[
 \Psi_0
 \le2Nd\exp\!\left(-\lambda\tau+\frac{\lambda^2}{2d}\right).
\]
Choosing \(\lambda=d\tau\) yields
\(\Psi_0\le2Nd\,e^{-d\tau^2/2}=1/2<1\).

The conditional value before assigning \(s_{t+1}\) is the average of the two
values obtained from \(s_{t+1}=+1\) and \(s_{t+1}=-1\).  At least one choice
therefore does not increase \(\Psi_t\).  Selecting that choice at each step
produces \(\Psi_d\le\Psi_0<1\), proving
\eqref{eq:nuc-derandomization-bound}.  Each of the \(d\) sign choices updates
\(Nd\) coordinate contributions, for \(O(Nd^2)\) arithmetic operations.  This
count treats exponential evaluations and comparisons as unit-cost; it is not a
bit-complexity claim.
\end{proof}

For \Cref{thm:nuc-achievability}, apply the theorem to the \(N=m+n\) vectors
in \(\R^K\) formed by the transposed rows
\((x_i^{(\varepsilon)})^\top\) and the columns
\(y_j^{(\varepsilon)}\), and use
\(W_\varepsilon=\diag(s)H^\top\).  The theorem is finite-set rather than
uniform over all vectors.  For offline factors---factors known in advance so
that preprocessing is amortized---a draw--verify--retry Hadamard procedure can
reduce expected preprocessing time: every accepted draw is certified, and the
success probability gives finite expected termination.  The deterministic
construction instead guarantees worst-case termination and supports
reproducible preprocessing without a random-number generator (RNG).  Improving
its runtime remains the
algorithmic gap in
\Cref{op:nuc-derandomization}.  All three flattening methods realize the same
peak-to-average principle.

\section{Range-to-energy equivalence and finite-set flatness}
\label{sec:nuc-uncertainty}

Up to this point, we have addressed two sources of quantization error:
factorization energy, minimized by Gram-matrix balancing in
\Cref{sec:nuc-balancing}, and scalar range, reduced by the orthogonal transforms
of \Cref{sec:nuc-achievability,sec:nuc-derandomization}.  Under the practical
inner-only restriction, however, one tension remains: can a single orthogonal
transform simultaneously flatten the rows of \(A\) and the columns of \(B\)
while leaving their outer indices fixed?  We quantify this tension through a
finite-set flatness constant.

\subsection{A finite-set flatness constant}

For a finite set \(\mathcal X\subset\R^K\) and \(W\in\Oo(K)\), we define the
finite-set flatness constant by
\begin{equation}
 \alpha(W;\mathcal X)
 =\sup_{\substack{x\in\mathcal X\\x\ne0}}
 K\frac{\ii{Wx}^{2}}{\tw{x}^{2}}.
 \label{eq:nuc-finite-uncertainty}
\end{equation}
The range--energy inequality gives \(1\le\alpha\le K\); if \(\alpha=1\), each
\(Wx\) with \(x\in\mathcal X\setminus\{0\}\) distributes its energy uniformly.
If the transposed rows of \(X\) belong to \(\mathcal X\), then
\begin{equation}
 \sum_i\ii{(XW^\top)_{i,:}}^2
 \le\frac{\alpha(W;\mathcal X)}{K}\fro X^2.
 \label{eq:nuc-finite-set-sum}
\end{equation}
Thus \(\alpha\) is the worst peak-to-average loss for the realized finite set,
not leverage-score or mutual coherence.

For the balanced sequence in \Cref{lem:nuc-balance}, write
\[
 X_\varepsilon=AT_\varepsilon,\qquad
 Y_\varepsilon=T_\varepsilon^{-1}B
\]
and choose \(W_\varepsilon\in\Oo(K)\).  We define the ratio-of-sums aggregates
\begin{equation}
\begin{split}
 \alpha_{A,\varepsilon}
 &=K\,
 \frac{\sum_i\ii{(X_\varepsilon W_\varepsilon)_{i,:}}^2}
      {\fro{X_\varepsilon}^2},\\
 \alpha_{B,\varepsilon}
 &=K\,
 \frac{\sum_j\ii{(W_\varepsilon^\top
                    Y_\varepsilon)_{:,j}}^2}
      {\fro{Y_\varepsilon}^2}.
\end{split}
\label{eq:nuc-aggregate-uncertainty}
\end{equation}
Each aggregate is at most the corresponding sup-based constant in
\eqref{eq:nuc-finite-uncertainty}, and both are at most \(K\).  The leading
error is exactly
\begin{equation}
 \Elead
 =\frac{c_A\alpha_{A,\varepsilon}
        +c_B\alpha_{B,\varepsilon}}{K}
  \fro{X_\varepsilon}^2\fro{Y_\varepsilon}^2.
 \label{eq:nuc-uncertainty-identity}
\end{equation}
With the effective aggregate
\[
 \alpha_{\mathrm{eff},\varepsilon}
 =\frac{c_A\alpha_{A,\varepsilon}
        +c_B\alpha_{B,\varepsilon}}{c_A+c_B},
\]
balancing gives
\begin{equation}
 \Elead(I,T_\varepsilon W_\varepsilon,I)
 =\frac{c_A+c_B}{K}\,
  \alpha_{\mathrm{eff},\varepsilon}
  \bigl(\nuc{AB}^2+o_\varepsilon(1)\bigr).
 \label{eq:nuc-effective-uncertainty}
\end{equation}
Because \(\alpha_{\mathrm{eff},\varepsilon}\le K\), the product
\(\alpha_{\mathrm{eff},\varepsilon}o_\varepsilon(1)\) remains
\(o_\varepsilon(1)\) for fixed dimensions.  Along this balanced sequence,
\(\alpha_{\mathrm{eff},\varepsilon}\) is bounded if and only if the inner-only
construction reaches the nuclear lower bound within a constant factor.

\subsection{Why uniform control needs redundancy}

Finite-set flattening chooses a rotation after seeing the rows and columns.
The certificate it produces is therefore instance-specific, and a natural
question is whether one rotation could instead be certified once, for every
input the factors will ever multiply.  The answer is no, and the obstruction
is one line: a square orthogonal transform cannot flatten every input, because
taking \(x=W^\top e_k\), the transpose of its \(k\)th row, gives \(Wx=e_k\),
whose maximum coordinate has magnitude one.  Uniform constant-level control
therefore requires a redundant representation, and the price of redundancy is
quantified by the Kashin theory that this subsection develops.

To avoid collision with the matrix-product dimensions, let
\(F:\R^{n_f}\to\R^{d_f}\) satisfy \(FF^\top=I_{d_f}\).  Its columns form a
Parseval tight frame.  A coefficient vector \(a\) is a Kashin representation
of \(x\) at level \(\lambda_0\) if
\begin{equation}
 x=Fa,\qquad
 \ii a\le \lambda_0\frac{\tw x}{\sqrt{n_f}}.
 \label{eq:nuc-kashin-representation}
\end{equation}
This spreads every represented vector among \(n_f\) coefficients; it is the
frame form of Kashin's theorem \cite{kashin1977diameters,kashin}.  With
\(B_2^{d_f}=\{x:\tw x\le1\}\) and
\(B_\infty^{n_f}=\{a:\ii a\le1\}\), the all-input condition is equivalent to
\begin{equation}
 B_2^{d_f}\subseteq
 \frac{\lambda_0}{\sqrt{n_f}}F(B_\infty^{n_f}).
 \label{eq:nuc-kashin-inclusion}
\end{equation}
For \(E=\operatorname{range}(F^\top)\), define
\begin{equation}
 \Lambda_1(E)=
 \sup_{z\in E\setminus\{0\}}
 \frac{\sqrt{n_f}\,\tw z}{\lVert z\rVert_1}.
 \label{eq:nuc-lambda-one}
\end{equation}
Then \eqref{eq:nuc-kashin-inclusion} implies
\(\Lambda_1(E)\le\lambda_0\).  This follows by polarity in
\Cref{app:nuc-kashin}, which also recalls the weak uncertainty principle of
Lyubarskii and Vershynin.

The two regimes are now adjacent, and their contrast is the structural
conclusion of the inner-only analysis.  Finite-set flatness for one fixed
product costs a logarithmic factor and is achievable by a deterministic
construction (\Cref{thm:nuc-derandomization}).  Uniform control of every
input is impossible for any square orthogonal transform at any constant
level and therefore requires redundancy.  For suitable frame families with a
fixed redundancy ratio \(n_f/d_f>1\), it becomes achievable at a
dimension-independent level \(\lambda_0\).  Redundancy, not rotation
selection, is what converts an instance-specific certificate into a
universal one.  We now evaluate the finite-set quantities on controlled
instances;
\Cref{sec:nuc-open} returns to the structural questions raised here.

\section{Reproducible numerical evaluation}
\label{sec:nuc-experiments}

The theory gives an exact full-gauge benchmark and a worst-case logarithmic
inner-only construction.  We use controlled synthetic instances to check those
predictions, the output invariance of the lower bound, and the regularized
singular limit.  We compute all reported errors from the analytic leading
objective \(\Phi_\infty\); sampled dither is unnecessary because the
second-moment model gives the expected error exactly.

\subsection{Protocol}

We derive all random streams from base seed \(20260723\).  The main study uses
\(m=n=K=64\), singular values
\(\operatorname{geomspace}(1,0.15,K)\), and 50 independent trials per
full-rank family.  Gaussian QR generates every Haar-distributed orthogonal
factor.  The families are:
\begin{enumerate}
  \item \emph{Spiky}: diagonal square-root factors receive independent
  \(N(0,1.35^2)\) log-scales; eight random coordinates receive an additional
  signed shift of magnitude \(2.25\), after which we mean-center and clip the
  scales to \([-\log 25,\log 25]\).
  \item \emph{Dense diagonal gauge}: independent \(N(0,1.35^2)\) log-scales
  multiply Haar-oriented factors.
  \item \emph{Dense general gauge}: the same base factors receive
  \(Q\operatorname{diag}(\exp(\operatorname{linspace}(-2,2,K)))Q^\top\)
  for an independent Haar \(Q\).
\end{enumerate}
The diagonal method exactly minimizes the factorization-energy objective over
positive diagonal gauges, using scales proportional to
\((\lVert B_{k,:}\rVert_2^2/\lVert A_{:,k}\rVert_2^2)^{1/4}\).
The general balancing transform is computed from
\eqref{eq:nuc-regularized-transform} by symmetric eigendecompositions with an
eigenvalue floor of
\(2\mathord{\times}10^{-14}\max\{\tr(P)/K,1\}\).  We also consider a separate
singular case with
\(\rank(A)=\rank(B)=32\), latent coordinate supports \(1{:}32\) and
\(17{:}48\), and hence \(\rank(AB)=16\).  Its nonzero factor values are
\(\sqrt{\operatorname{geomspace}(1,0.2,32)}\), and every
\(\varepsilon\in\{10^{-1},\ldots,10^{-10}\}\) is used without adaptive
stopping.

For each family, we compare the original factors, diagonal balancing plus
randomized Hadamard flattening, Gram-matrix balancing plus either Haar or
deterministic Hadamard flattening, and the full-gauge SVD--Hadamard construction
from \Cref{thm:nuc-full-gauge}.  The last method is the exact reference at
\(K=64\).  To isolate the flattening methods, we use one fixed balanced
instance with 100 independent Haar and randomized-Hadamard draws and one
deterministic certificate.  For this comparison, the reported peak-to-average
statistic is the finite-set flatness constant from
\eqref{eq:nuc-finite-uncertainty}.

The normalized objective-to-floor ratio is algebraically independent of a
common bit width, so duplicating it at 4 and 8 bits would add no experiment.
In the raw validation, we nevertheless evaluate this ratio at both bit widths
and require the two values to agree numerically.

\subsection{Generated results}

% Generated by experiments/reproduce.py; do not edit.
\providecommand{\ExperimentSeed}{20260723}
\providecommand{\NuclearExperimentDimension}{64}
\providecommand{\NuclearExperimentTrials}{50}
\providecommand{\NuclearFlatteningDraws}{100}
\providecommand{\NuclearDimensionScalingTrials}{8}
\providecommand{\NuclearDimensionScalingLargestK}{1024}
\providecommand{\NuclearFullGaugeMedianRatio}{1.000}
\providecommand{\NuclearSpikyOriginalMedianRatio}{4.538\times 10^{4}}
\providecommand{\NuclearSpikyDeterministicMedianRatio}{1.000}
\providecommand{\NuclearDenseDiagonalOriginalMedianRatio}{1.141\times 10^{4}}
\providecommand{\NuclearDenseDiagonalBalancedMedianRatio}{6.307}
\providecommand{\NuclearGeneralOriginalMedianRatio}{355.643}
\providecommand{\NuclearGeneralBalancedMedianRatio}{6.389}
\providecommand{\NuclearMaximumRelativeFloorSpread}{9.442\times 10^{-16}}
\providecommand{\NuclearSingularSmallestEpsilon}{1.000\times 10^{-10}}
\providecommand{\NuclearSingularLargestCondition}{8.125\times 10^{4}}
\providecommand{\NuclearSingularSmallestExcess}{2.991\times 10^{-5}}
\providecommand{\NuclearHaarMedianRatio}{6.975}
\providecommand{\NuclearRandomHadamardMedianRatio}{6.943}
\providecommand{\NuclearDeterministicHadamardRatio}{6.044}
\providecommand{\NuclearDeterministicHadamardCoherence}{10.116}
\providecommand{\NuclearScalingSmallestHaarRatio}{4.480}
\providecommand{\NuclearScalingLargestHaarRatio}{11.933}
\providecommand{\NuclearScalingLargestRandomHadamardRatio}{11.964}

The full-gauge baseline yields a median ratio of
\(\NuclearFullGaugeMedianRatio\) in every family.  For the spiky family, we
reduce the median from \(\NuclearSpikyOriginalMedianRatio\) in the original
factorization to \(\NuclearSpikyDeterministicMedianRatio\) by applying diagonal
balancing and deterministic Hadamard flattening.  The exact \(1.000\) follows
algebraically: the diagonal output leaves the balanced factors SVD-aligned, and
the Hadamard rotation makes every relevant row and column flat, so outer
rotations are unnecessary in the construction of
\Cref{thm:nuc-full-gauge}.  For the dense diagonal-gauge family the medians are
\(\NuclearDenseDiagonalOriginalMedianRatio\) (original) and
\(\NuclearDenseDiagonalBalancedMedianRatio\) (general inner balancing plus
deterministic Hadamard flattening); for the dense general-gauge family they are
\(\NuclearGeneralOriginalMedianRatio\) and
\(\NuclearGeneralBalancedMedianRatio\), respectively.

Across independently refactorized copies of the same output, the maximum
relative spread of the nuclear lower bound is
\(\NuclearMaximumRelativeFloorSpread\).  This is a numerical invariance check,
not a statistical estimate.  On the fixed flattening instance, median ratios
over \(\NuclearFlatteningDraws\) draws are \(\NuclearHaarMedianRatio\) for Haar
and \(\NuclearRandomHadamardMedianRatio\) for randomized Hadamard.  The
deterministic construction gives a ratio of
\(\NuclearDeterministicHadamardRatio\) and a maximum finite-set flatness of
\(\NuclearDeterministicHadamardCoherence\).

\begin{table}[t]
  \centering
  \small
  \caption{Leading objective divided by
  \((c_A+c_B)\nuc{AB}^{\,2}/K\) at \(m=n=K=64\).  Entries show medians and
  interquartile ranges over 50 trials per family, all derived from base seed
  \(20260723\).  In the table, SVD denotes singular-value decomposition, H a
  Hadamard transform, and IQR the interquartile range.}
  \label{tab:nuc-generated-family-results}
  % Generated by experiments/reproduce.py; do not edit.
\begin{tabular}{llrr}
\toprule
Family & Method & Median & IQR \\
\midrule
Spiky & Full gauge (SVD + Hadamard) & $1.000$ & $[1.000, 1.000]$ \\
Spiky & Original & $4.538\times 10^{4}$ & $[2.819\times 10^{4}, 7.450\times 10^{4}]$ \\
Spiky & Diagonal + randomized H & $1.000$ & $[1.000, 1.000]$ \\
Spiky & General inner + Haar & $6.849$ & $[6.746, 7.071]$ \\
Spiky & General inner + deterministic H & $1.000$ & $[1.000, 1.000]$ \\
Dense diagonal gauge & Full gauge (SVD + Hadamard) & $1.000$ & $[1.000, 1.000]$ \\
Dense diagonal gauge & Original & $1.141\times 10^{4}$ & $[6357.632, 2.516\times 10^{4}]$ \\
Dense diagonal gauge & Diagonal + randomized H & $6.915$ & $[6.785, 6.989]$ \\
Dense diagonal gauge & General inner + Haar & $6.835$ & $[6.716, 6.971]$ \\
Dense diagonal gauge & General inner + deterministic H & $6.307$ & $[6.241, 6.424]$ \\
Dense general gauge & Full gauge (SVD + Hadamard) & $1.000$ & $[1.000, 1.000]$ \\
Dense general gauge & Original & $355.643$ & $[343.700, 363.076]$ \\
Dense general gauge & Diagonal + randomized H & $331.807$ & $[324.801, 340.266]$ \\
Dense general gauge & General inner + Haar & $6.911$ & $[6.800, 7.019]$ \\
Dense general gauge & General inner + deterministic H & $6.389$ & $[6.279, 6.445]$ \\
\bottomrule
\end{tabular}

\end{table}

\begin{figure}[t]
  \centering
  \includegraphics[width=0.72\linewidth]{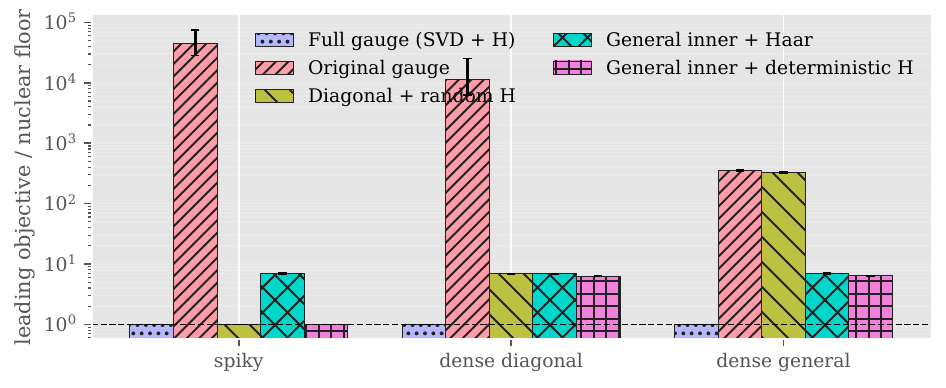}

  \includegraphics[width=0.72\linewidth]{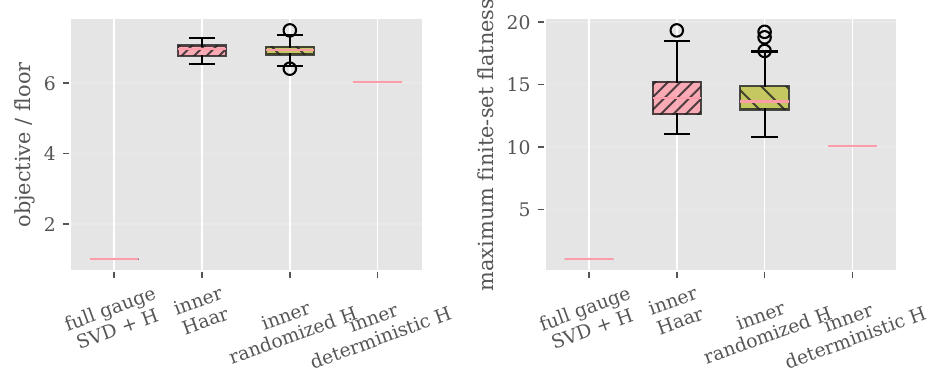}
  \caption{\emph{Top:} Objective-to-floor ratios for the three fixed-output
  families at \(m=n=K=64\); bars and whiskers show medians and interquartile
  ranges over 50 trials.  \emph{Bottom:} Leading ratios and maximum finite-set
  flatness on one balanced instance; boxes summarize 100 Haar and
  randomized-Hadamard draws, and points show the deterministic certificate.
  In both panels, the full-gauge singular-value-decomposition (SVD)--Hadamard
  (H) reference remains at one.}
  \label{fig:nuc-method-comparisons}
\end{figure}

\subsection{Dimension scaling and conditioning}

To separate the two theorem regimes, we also use non-square dimensions
\[
 K\in\{16,32,64,128,256,1024\},\qquad m=3K/2,\qquad n=5K/4.
\]
At each dimension, eight independent Gaussian-QR Stiefel pairs form
balanced factors with the same geometric singular-value schedule.  The median
inner-Haar ratio rises from \(\NuclearScalingSmallestHaarRatio\) to
\(\NuclearScalingLargestHaarRatio\), whereas the full-gauge baseline remains
one; randomized Hadamard follows the same trend.

\begin{figure}[t]
  \centering
  \includegraphics[height=0.29\linewidth]{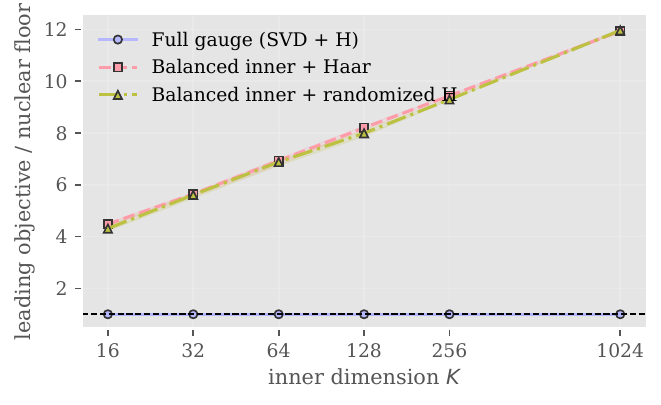}\hfill
  \includegraphics[height=0.29\linewidth]{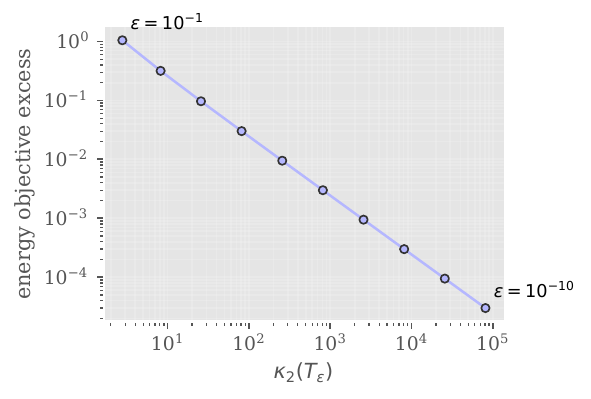}
  \caption{\emph{Left:} Dimension sweep for non-square balanced factors; curves
  show medians and interquartile ranges over eight trials per dimension.
  Inner-only flattening grows slowly while the full-gauge construction stays at
  one.  \emph{Right:} Regularization sweep for one rank-deficient
  \(m=n=K=64\) instance, showing
  \(\lVert AT_\varepsilon\rVert_F^2
    \lVert T_\varepsilon^{-1}B\rVert_F^2/\nuc{AB}^{\,2}-1\)
  against \(\kappa_2(T_\varepsilon)\).  Legend abbreviations use SVD for
  singular-value decomposition and H for a Hadamard transform.}
  \label{fig:nuc-scaling-and-conditioning}
\end{figure}

For the rank-deficient case, the smallest regularizer
\(\varepsilon=\NuclearSingularSmallestEpsilon\) yields
\(\kappa_2(T_\varepsilon)=\NuclearSingularLargestCondition\) and a relative
energy-objective excess of \(\NuclearSingularSmallestExcess\).  This
regularization sweep illustrates the limiting construction in
\Cref{lem:nuc-balance} for this instance.

\paragraph{Computational regime.}
For static factors, the SVD and balancing work can be performed offline.
A dense inner transform costs \(O((m+n)K^2)\) to apply, while a Hadamard
transform costs
\(O((m+n)K\log K)\).  Dense outer rotations and their undo cost
\(O(mn(m+n))\) at the output unless they are fused or replaced by structured
maps; the inner-only construction is therefore the closer model when free
indices carry operational meaning.  If transform conditioning or application
cost is constrained, the relevant practical problem includes the constraint
\(\kappa(T)\le\kappa_{\max}\) or a structured-transform restriction.  The
singular sweep shows one instance of that tradeoff.

\subsection{Reproduction and validation}

The source repository is
\url{https://github.com/piyush314/gauge-floors}.  The command
\begin{center}
\small\texttt{python -m experiments.reproduce nuclear}
\end{center}
regenerates the CSV records, summary JSON, TeX tables and macros, and vector
figures.  The file \path{experiments/results/nuclear_manifest.json} records the
environment versions, input seed, and SHA-256 hashes of generated artifacts;
\path{experiments/README.md} documents the exact generators and numerical
tolerances.

The run fails if any normalized objective falls below one beyond floating-point
tolerance, if the range objective disagrees with
\eqref{eq:nuc-uncertainty-identity}, if a product changes under a gauge, or if
the deterministic Hadamard certificate violates
\eqref{eq:nuc-derandomization-bound}.  These checks connect the artifacts to
the proved invariants; the remaining numerical results characterize the stated
synthetic families.

\section{Open problems}
\label{sec:nuc-open}

The experiments numerically validate the implementation of the proved
invariants; the remaining gaps are
theoretical.  Over the full gauge, the leading-error lower bound is attained at
Hadamard dimensions and is within a factor of two in general.  The open
questions concern inner-only flatness, faster certification, uniform frames,
and the boundary of the stochastic scalar model.

\begin{openproblem}[Worst-case normalized range gap for inner-only transforms]
\label{op:nuc-gap}
For \(AB\ne0\) and fixed positive \(c_A,c_B\), define
\begin{equation}
 \Gamma(A,B)=
 \frac{\displaystyle\inf_{T\in\GL(K)}\Phi_\infty(T)}
 {\displaystyle\frac{c_A+c_B}{K}\nuc{AB}^{\,2}}.
 \label{eq:nuc-normalized-gap}
\end{equation}
\Cref{eq:nuc-two-sided} gives
\(1\le\Gamma(A,B)\le C\log(K(m+n))\).  Over the full gauge,
\Cref{thm:nuc-full-gauge} bounds the corresponding ratio by two; this question
imposes \(U_L=U_R=I\).  Is the supremum of \(\Gamma(A,B)\) bounded by a
dimension-independent constant, or what is its sharp growth in \(m,n,K\)?
Can a nonorthogonal construction improve the finite-set logarithm?
\end{openproblem}

\begin{openproblem}[Near-linear deterministic flattening]
\label{op:nuc-derandomization}
Can one achieve the guarantee of \Cref{thm:nuc-derandomization}
deterministically in \(\widetilde O(Nd)\) time instead of \(O(Nd^2)\), where
\(\widetilde O\) hides polylogarithmic factors?  Such a result would narrow the
preprocessing gap between deterministic certification and randomized Hadamard
flattening.
\end{openproblem}

\begin{openproblem}[Partial-Fourier uniform Kashin frames]
\label{op:nuc-kashin}
Fix a redundancy ratio \(n_f/d_f>1\).  Can one choose \(d_f\) rows \(\Omega\)
of the normalized \(n_f\times n_f\) discrete Fourier matrix so that
\[
 E_\Omega=\operatorname{range}(F_\Omega^*)
 \quad\text{satisfies}\quad
 \Lambda_1(E_\Omega)\le C
\]
for a dimension-independent \(C\)?  Equivalently, can these partial-Fourier
frames have a dimension-independent Kashin level in the normalization of
\Cref{sec:nuc-uncertainty}?  This is the classical partial-Fourier
\(\Lambda_1\) problem discussed in the Kashin-frame literature
\cite{kashin}.
\end{openproblem}

\begin{openproblem}[Finite-rate converse beyond the model]
\label{op:nuc-finite-rate}
Which instance invariant governs finite-rate error after clipping, correlated
rounding, or vector and nested-lattice quantization?  Is the nuclear norm still
one component of a sharp converse, and what additional dependence on codebook
geometry, overload, or factor reuse is necessary?
\end{openproblem}

\section{Conclusion}
\label{sec:nuc-conclusion}

The output nuclear norm sets a gauge-invariant lower bound on the leading
expected squared product error for matrix factors under independent per-vector
dither noise.  Balancing the Gram matrices closes the factorization energy
loss; regularization covers singular factors.  Over the full gauge, an
SVD-aligned construction makes the infimum equal to this bound at Hadamard
dimensions and comes within a factor of two for every dimension \(K\).
Restricted to inner transforms, finite-set flattening reaches the bound within
a logarithmic factor.  Deterministic Hadamard sign selection gives the same
guarantee at power-of-two dimensions.  The exact bilinear term yields the
nonasymptotic complete-error lower bound
\eqref{eq:nuc-full-error-floor}.

The lower bound and matching constructions turn the output nuclear norm into
both an invariant obstruction and a constructive target.  Together, they
characterize the full-gauge optimum up to a
universal constant for every dimension and isolate finite-set flatness as the
remaining cost when meaningful free indices must be preserved.

This characterization gives transform designers two diagnostics that require
no simulated quantization: compare a candidate with the output-only nuclear
floor, then attribute the residual gap to factorization energy and finite-set
flatness.  The Kashin interpretation also explains why controlling a fixed
collection of rows and columns is easier than controlling every possible
input: uniform control generally needs redundant representations.  Extending
these certificates to finite grids with clipping, correlated errors, and
application-constrained transforms is therefore the natural next step.

\section*{Acknowledgments}

The authors acknowledge support from the 2026 Laboratory Directed Research and
Development (LDRD) FORSEE initiative, ``CCSD Core: Foundational Research for
Smart Extreme-scale Ecosystems,'' at Oak Ridge National Laboratory.

Generative AI tools were used for LaTeX formatting, copy-editing, figure
generation, and for assistance with scripts used in numerical checks. They
were also used to help locate and organize related literature. All theorems,
proofs, citations, and claims of novelty were checked by the author, who
takes full responsibility for the content of this paper.

{\footnotesize
\bibliographystyle{plain}
\bibliography{references}

@article{qmm,
  title        = {Contraction-Gauge Preconditioning for Quantized Matrix Multiplication},
  author       = {Sao, Piyush and Miniskar, Narasinga and Valero-Lara, Pedro and Teranishi, Keita and Seal, Sudip},
  journal      = {arXiv preprint arXiv:2607.18745},
  year         = {2026},
  doi          = {10.48550/arXiv.2607.18745},
  url          = {https://arxiv.org/abs/2607.18745}
}

@unpublished{metriccompanion,
  title        = {Metric Discrepancy for Product-Aware Deterministic Matrix Rounding},
  author       = {Sao, Piyush and Miniskar, Narasinga and Valero-Lara, Pedro and Teranishi, Keita and Seal, Sudip},
  note         = {Companion manuscript},
  year         = {2026}
}

@article{graystockham,
  title        = {Dithered Quantizers},
  author       = {Gray, Robert M. and Stockham, Thomas G.},
  journal      = {IEEE Transactions on Information Theory},
  volume       = {39},
  number       = {3},
  pages        = {805--812},
  year         = {1993},
  doi          = {10.1109/18.256489}
}

@article{grayneuhoff,
  title        = {Quantization},
  author       = {Gray, Robert M. and Neuhoff, David L.},
  journal      = {IEEE Transactions on Information Theory},
  volume       = {44},
  number       = {6},
  pages        = {2325--2383},
  year         = {1998},
  doi          = {10.1109/18.720541}
}

@article{kashin,
  title        = {Uncertainty Principles and Vector Quantization},
  author       = {Lyubarskii, Yurii and Vershynin, Roman},
  journal      = {IEEE Transactions on Information Theory},
  volume       = {56},
  number       = {7},
  pages        = {3491--3501},
  year         = {2010},
  doi          = {10.1109/TIT.2010.2048458}
}

@article{kashin1977diameters,
  title        = {Diameters of Some Finite-Dimensional Sets and Classes of Smooth Functions},
  author       = {Kashin, B. S.},
  journal      = {Mathematics of the USSR-Izvestiya},
  volume       = {11},
  number       = {2},
  pages        = {317--333},
  year         = {1977},
  doi          = {10.1070/IM1977v011n02ABEH001719},
  url          = {https://www.mathnet.ru/eng/im1805}
}

@article{raghavan1988probabilistic,
  title        = {Probabilistic Construction of Deterministic Algorithms: Approximating Packing Integer Programs},
  author       = {Raghavan, Prabhakar},
  journal      = {Journal of Computer and System Sciences},
  volume       = {37},
  number       = {2},
  pages        = {130--143},
  month        = oct,
  year         = {1988},
  doi          = {10.1016/0022-0000(88)90003-7},
  url          = {https://www.sciencedirect.com/science/article/pii/0022000088900037}
}

@article{moore1981balanced,
  title        = {Principal Component Analysis in Linear Systems: Controllability, Observability, and Model Reduction},
  author       = {Moore, Bruce C.},
  journal      = {IEEE Transactions on Automatic Control},
  volume       = {26},
  number       = {1},
  pages        = {17--32},
  year         = {1981},
  doi          = {10.1109/TAC.1981.1102568}
}

@inproceedings{srebro2005rank,
  title        = {Rank, Trace-Norm and Max-Norm},
  author       = {Srebro, Nathan and Shraibman, Adi},
  booktitle    = {Learning Theory: 18th Annual Conference on Learning Theory, {COLT} 2005},
  series       = {Lecture Notes in Computer Science},
  volume       = {3559},
  pages        = {545--560},
  publisher    = {Springer},
  year         = {2005},
  doi          = {10.1007/11503415_37}
}

@article{ordentlich2024optimal,
  title        = {Optimal Quantization for Matrix Multiplication},
  author       = {Ordentlich, Or and Polyanskiy, Yury},
  journal      = {IEEE Transactions on Information Theory},
  volume       = {72},
  number       = {3},
  pages        = {1943--1972},
  year         = {2026},
  doi          = {10.1109/TIT.2025.3649596},
  note         = {arXiv:2410.13780}
}

@article{ordentlich2026highrate,
  title        = {High-Rate Quantized Matrix Multiplication {I}},
  author       = {Ordentlich, Or and Polyanskiy, Yury},
  journal      = {arXiv preprint arXiv:2601.17187},
  year         = {2026},
  doi          = {10.48550/arXiv.2601.17187},
  note         = {Version 2},
  url          = {https://arxiv.org/abs/2601.17187}
}

@article{ordentlich2026highrate2,
  title        = {High-Rate Quantized Matrix Multiplication {II}},
  author       = {Ordentlich, Or and Polyanskiy, Yury},
  journal      = {arXiv preprint arXiv:2605.13768},
  year         = {2026},
  doi          = {10.48550/arXiv.2605.13768},
  note         = {Version 2},
  url          = {https://arxiv.org/abs/2605.13768}
}

@inproceedings{xiao2023smoothquant,
  title        = {{SmoothQuant}: Accurate and Efficient Post-Training Quantization for Large Language Models},
  author       = {Xiao, Guangxuan and Lin, Ji and Seznec, Mickael and Wu, Hao and Demouth, Julien and Han, Song},
  booktitle    = {Proceedings of the 40th International Conference on Machine Learning},
  series       = {Proceedings of Machine Learning Research},
  volume       = {202},
  year         = {2023}
}

@inproceedings{ashkboos2024quarot,
  title        = {{QuaRot}: Outlier-Free 4-Bit Inference in Rotated {LLM}s},
  author       = {Ashkboos, Saleh and Mohtashami, Amirkeivan and Croci, Maximilian L. and Li, Bo and Cameron, Pashmina and Jaggi, Martin and Alistarh, Dan and Hoefler, Torsten and Hensman, James},
  booktitle    = {Advances in Neural Information Processing Systems},
  year         = {2024},
  note         = {arXiv:2404.00456}
}

@inproceedings{liu2025spinquant,
  title        = {{SpinQuant}: {LLM} Quantization with Learned Rotations},
  author       = {Liu, Zechun and Zhao, Changsheng and Fedorov, Igor and Soran, Bilge and Choudhary, Dhruv and Krishnamoorthi, Raghuraman and Chandra, Vikas and Tian, Yuandong and Blankevoort, Tijmen},
  booktitle    = {International Conference on Learning Representations},
  year         = {2025},
  url          = {https://openreview.net/forum?id=ogO6DGE6FZ}
}

@inproceedings{ma2024affinequant,
  title        = {{AffineQuant}: Affine Transformation Quantization for Large Language Models},
  author       = {Ma, Yuexiao and Li, Huixia and Zheng, Xiawu and Ling, Feng and Xiao, Xuefeng and Wang, Rui and Wen, Shilei and Chao, Fei and Ji, Rongrong},
  booktitle    = {International Conference on Learning Representations},
  year         = {2024},
  url          = {https://arxiv.org/abs/2403.12544}
}

@inproceedings{wush,
  title        = {{WUSH}: Near-Optimal Adaptive Transforms for {LLM} Quantization},
  author       = {Chen, Jiale and Egiazarian, Vage and Castro, Roberto L. and Hoefler, Torsten and Alistarh, Dan},
  booktitle    = {Proceedings of the 43rd International Conference on Machine Learning},
  series       = {Proceedings of Machine Learning Research},
  volume       = {306},
  year         = {2026},
  address      = {Seoul, South Korea},
  note         = {arXiv:2512.00956v3; PMLR pagination pending},
  url          = {https://arxiv.org/abs/2512.00956v3}
}
}

\appendix
\section{Kashin duality and the weak uncertainty principle}
\label{app:nuc-kashin}

Polarity---the duality between a convex body and the linear functionals bounded
by one on it---makes the geometry behind
\eqref{eq:nuc-kashin-inclusion} explicit.  The support functions are
\[
 h_{B_2^{d_f}}(y)=\tw y,\qquad
 h_{F(B_\infty^{n_f})}(y)=\lVert F^\top y\rVert_1.
\]
Consequently, \eqref{eq:nuc-kashin-inclusion} is equivalent to
\[
 \tw y\le
 \frac{\lambda_0}{\sqrt{n_f}}\lVert F^\top y\rVert_1.
\]
Because \(F^\top\) is an isometry onto
\(E=\operatorname{range}(F^\top)\), this becomes
\[
 \lVert z\rVert_1\ge
 \frac{\sqrt{n_f}}{\lambda_0}\tw z,\qquad z\in E,
\]
which is \(\Lambda_1(E)\le\lambda_0\).

Lyubarskii and Vershynin obtain such representations from a weak uncertainty
principle \cite{kashin}.  For \(\delta,\eta\in(0,1)\), suppose
\[
 |\operatorname{supp}z|\le\delta n_f
 \quad\Longrightarrow\quad
 \tw{Fz}\le\eta\tw z.
\]
Their iterative truncation procedure clips large coefficients, represents the
residual, and repeats.  The residual contracts geometrically, producing a
Kashin representation with level
\(\lambda_0=(1-\eta)^{-1}\delta^{-1/2}\).

\end{document}